\documentclass[conference]{IEEEtran}
\usepackage{graphicx}
\usepackage{amsmath}
\usepackage{amsfonts}
\usepackage{amssymb}
\usepackage{cite}
\usepackage{color}
\usepackage{multirow}
\usepackage{enumerate}
\usepackage{tabularx}
\usepackage{bbm}
\usepackage{bm,amsthm}
\usepackage{comment}
\newcolumntype{L}[1]{>{\raggedright\arraybackslash}p{#1}}
\newcolumntype{C}[1]{>{\centering\arraybackslash}p{#1}}
\newcolumntype{R}[1]{>{\raggedleft\arraybackslash}p{#1}}
\usepackage{arydshln}
\usepackage{times, epsfig}
\usepackage{subfig}

\definecolor{kkk}{RGB}{165,42,42}

\newlength{\figwidth}
\begin{document}
	
\setlength{\pdfpagewidth}{8.5in}
\setlength{\pdfpageheight}{11in}
\title{Outage Analysis of $2\times2 $ MIMO-MRC in Correlated Rician Fading}
\author{
\IEEEauthorblockN{
Prathapasinghe~Dharmawansa\IEEEauthorrefmark{1}, 
Kumara~Kahatapitiya\IEEEauthorrefmark{1},
Saman~Atapattu\IEEEauthorrefmark{3}, and 
Chintha~Tellambura\IEEEauthorrefmark{4}
}
\IEEEauthorblockA{
\IEEEauthorrefmark{1}Department of Electronic and Telecommunication Engineering, University of Moratuwa, Moratuwa, Sri Lanka \\
\IEEEauthorrefmark{3}Department of Electrical and Electronic Engineering, University of Melbourne, Victoria, Australia \\
\IEEEauthorrefmark{4}Department of Electrical and Computer Engineering, University of Alberta, Edmonton, Canada
}
}
	
\maketitle
\begin{abstract}
This paper addresses one of the classical problems in random matrix theory-- finding the distribution of  the maximum eigenvalue of the correlated Wishart unitary ensemble. In particular, we derive  a new exact expression for  the cumulative distribution function (c.d.f.) of the maximum eigenvalue of a $2\times 2$ correlated non-central Wishart matrix with rank-$1$ mean. By using this  new result, we derive an exact analytical expression for the outage probability of  $2\times 2$ multiple-input multiple-output maximum-ratio-combining (MIMO-MRC) in   Rician fading with transmit correlation and a strong line-of-sight (LoS) component (rank-$1$ channel mean). We also show  that the outage performance is affected by the relative alignment of the eigen-spaces of the mean and correlation matrices.  In general, when the LoS path aligns with the least eigenvector of the correlation matrix, in the {\it high} transmit signal-to-noise ratio  (SNR) regime, the outage gradually improves with the increasing correlation. Moreover, we show that  as $K$ (Rician factor) grows large, the outage event can be approximately characterized by the c.d.f. of a certain Gaussian random variable. 
\end{abstract}


\IEEEpeerreviewmaketitle

\section{Introduction}
Multiple-input multiple-output (MIMO) systems are critical  for 3G (third generation)  and current 4G (fourth generation)   wireless networks\cite{Tse2003it}. They are  also a key enabler  for 5G (fifth generation)  and beyond wireless networks to  accommodate the ever-increasing data demand   \cite{Marzetta2010twc}.
For example,   in cell-free massive MIMO, where a  large number of access points are distributed over a wide area, users and access point may have only one or  two  antennas \cite{Ngo2017twc}.  Therefore, the  simplest $2\times 2$ MIMO channel becomes the most common channel and hence  is   of paramount importance  for small-cell  wireless networks\cite{Diab2016}. Moreover, for  such small-size MIMO channels, accurate channel state information (CSI) estimation is fully  viable, which enables the use of  MIMO maximum-ratio-combining (MRC) receivers. 

The performance of MIMO-MRC has been extensively analyzed   for an  arbitrary number of transmit and receive antennas for  different channel models (\cite{dighe2003analysis,rao2003tcom,mckay2007performance,jin2008mimo,wu2016asymptotic} and references therein). However, exact analytical characterization of the MIMO-MRC over the correlated Rician channel remains elusive. This challenge arises because  the  $ 2 \times 2 $ correlated non-central Wishart matrix does not admit analytically tractable joint eigenvalue density \cite{ratnarajah2003topics}. 
While general MIMO systems are considered  in the literature, 
only  few studies have focused on small-scale practical MIMO architectures \cite{Cioffi2008it, McKay2017tvt, McKay2018vt, Paulraj2007asilomar}. For instance, \cite{Cioffi2008it} derives the channel capacity of $2 \times 2$ or $2 \times 3$ MIMO channel  for Nakagami-$m$ fading. 
In  \cite{McKay2017tvt}, the statistical properties of the Gram matrix ${\bf W = HH^\dagger}$, where ${\bf H}$ is a $2 \times 2$ complex central Gaussian matrix whose elements have arbitrary variances, have been investigated, resulting in exact distributions of ${\bf W}$ and its eigenvalues.  In \cite{McKay2018vt}, the exact and asymptotic  largest eigenvalue distributions of ${\bf W}$ are derived when ${\bf H}$ is a complex Gaussian matrix with unequal variances in the real and imaginary parts of its entries, or equivalently   ${\bf H}$ belongs to the non-circularly-symmetric Gaussian subclass. These results have been then leveraged to analyze the outage performance of multi-antenna systems with MRC over Nakagami-$q$ (Hoyt) fading.


To the best of our knowledge, \textbf{no  exact performance analysis is available} for  the important case of $2 \times 2$ MIMO  correlated Rician fading channel with   rank-$1$ channel mean, i.e., a  strong line-of-sight (LoS) path exists between transmitter and receiver \cite{nabar2005diversity,kang2006capacity}. 
To develop exact analytical results, we must first statistically characterize the maximum eigenvalue of the  correlated non-central Wishart matrix. However, the distribution of this  maximum eigenvalue  remains an open problem in both wireless  and even wider statistics literature. 
The main technical challenge is  that the joint eigenvalue density of the correlated non-central Wishart matrix has no tractable representation \cite{ratnarajah2003topics}. Since  the invariant polynomial representation for the joint eigenvalue density given in \cite{ratnarajah2003topics} is not amenable to further manipulations, 
the direct technique of obtaining marginal densities by integrating the joint density of the eigenvalues over a suitable multi-dimensional region
cannot be applied in the problem at hand.  To circumvent this difficulty, here we follow an alternative approach \cite{dharmawansa2011extreme},
where  the matrix variate density is directly integrated  instead of the joint eigenvalue density, in order to find the  maximum-eigenvalue distribution. 

Specifically,  we first derive  an exact c.d.f. of the maximum eigenvalue of a $2\times 2$ complex correlated none-central Wishart matrix with rank-1 mean, which is the key contribution of this paper. Some recent results on the maximum eigenvalue are symbolic,  and not amenable to further processing, e.g., \cite{dharmawansa2011extreme,wu2016asymptotic}.
We thus believe this to be  the first tractable exact result on the maximum eigenvalue of such matrices. Subsequently, we also provide an exact expression for the outage probability. 
Additionally, we also characterize the effect of the Rician factor $K$ on the outage probability for different signal-to-noise ratio (SNR) regimes and  establish stochastic convergence limit of  outage probability for large $K$ values. Moreover, the effect of the eigenspaces of the mean and the correlation matrices on the outage probability is  discussed in detail.

The following notation is used throughout this paper.  The superscript $(\cdot)^\dagger$ indicates the Hermitian-transpose, and $(\cdot)^T$ stands for the matrix transpose. We use $\text{det}(\cdot)$ to represent the determinant of a square matrix, $\text{tr}(\cdot)$ to represent trace, and $\text{etr}(\cdot)$ stands for $\exp\left(\text{tr}(\cdot)\right)$. Positive definiteness  of a square matrix $\mathbf{A}$ is represented by $\mathbf{A}\succ \mathbf{0}$, and $\mathbf{A}\succ \mathbf{B}$ denotes $\mathbf{A}-\mathbf{B}\succ \mathbf{0}$. The square root of a positive definite matrix $\mathbf{G}$ is denoted by $\mathbf{G}^{\frac{1}{2}}$ and $\text{diag}\{s_1, s_2\}$ denotes a $2\times 2$ diagonal matrix with the real diagonal entries $s_1$ and $s_2$. We use $\lambda_{\max}(\mathbf{A})$ and $\lambda_{\min}(\mathbf{A})$ to denote, respectively, the maximum and the minimum eigenvalues of a square matrix $\mathbf{A}$. The real and imaginary parts, modulus and conjugate  of a complex number $z$ are denoted, respectively, by $\Re(z),$  $\Im(z)$,  $|z|$ and  $z^*$. The Euclidean norm of a vector $\mathbf{w}$ is denotes by $||\mathbf{w}||$. $\lfloor x\rfloor$ denotes the floor function, defined as $\lfloor{x}\rfloor=\max \left\{m\in \mathbb{Z}|m\leq x\right\}$. Finally, the union of two measurable sets $\mathcal{R}_1$ and $\mathcal{R}_2$ is denoted by $\mathcal{R}_1\cup \mathcal{R}_2$. 

\section{CDF of the Maximum Eigenvalue}
\label{sec:cdfmaxeigen}
This section derives the  new expression for the distribution  of the maximum eigenvalue of a $2\times 2$ correlated complex non-central Wishart matrix with rank-1  mean matrix. Before proceeding with the derivations, we  present some fundamental statistical characteristics of the complex correlated non-central Wishart matrix.

\newtheorem{theorem}{Definition}
\begin{theorem}\label{Wishart_Def}
Let $\mathbf{X}$ be an $n\times m$ ($n \geq m$) complex Gaussian random matrix distributed as $\mathcal{CN}_{n,m}\left(\boldsymbol{\Upsilon},\mathbf{I}_n\otimes \boldsymbol{\Psi }\right)$, where $\boldsymbol{\Psi}\in \mathbb{C}^{m \times m}\succ\mathbf{0}$ and $\boldsymbol{\Upsilon}\in \mathbb{C}^{n\times m}$. Then $\mathbf{W}=\mathbf{X}^\dagger\mathbf{X}$ has a complex non-central Wishart distribution $\mathcal{W}_m\left(n,\boldsymbol{\Psi},\boldsymbol{\Theta}\right)$ with the density function \cite{james1964distributions}
\begin{align}
\label{wishart}
f\left(\mathbf{W}\right)& =
\frac{\text{etr}\left(-\boldsymbol{\Theta}\right)\det^{n-m}\left(\mathbf{W}\right)}{\tilde{\Gamma}_m(n)\det^n(\boldsymbol{\Psi})}
\; \nonumber\\
&\qquad\text{etr}\left(-\boldsymbol{\Psi}^{-1}\mathbf{W}\right) {}_0\tilde{F}_1\left(n;\boldsymbol{\Theta}\boldsymbol{\Psi}^{-1}\mathbf{W}\right)
\end{align}
where $\boldsymbol{\Theta}=\boldsymbol{\Psi}^{-1}\boldsymbol{\Upsilon}^\dagger\boldsymbol{\Upsilon}$ is the \textit{non-centrality} parameter, 
${}_0\tilde{F}_1(\cdot;\cdot)$ denotes Bessel type complex hypergeometric function of matrix argument and the complex multivariate gamma function is defined as
$
\tilde{\Gamma}_m(n)\stackrel{ \Delta}{=}\pi^{\frac{m(m-1)}{2}}\prod_{j=1}^{m}\Gamma(n-j+1)
$
with $\Gamma(\cdot)$ being the gamma function.
\end{theorem}
Next we define the joint eigenvalue distribution of $\mathbf{W}$ which consists of invariant polynomials due to Davis \cite{davis1979invariant, davis1981construction}.
\newtheorem{cor}{Corollary}
\begin{cor}
The joint density of ordered eigenvalues $\lambda_1>\lambda_2>......>\lambda_m>0$, of the complex non-central Wishart matrix $\mathbf{W}$ is given by \cite[Eq. 5.4]{ratnarajah2003topics}
\begin{equation}
\label{eigenpdf}
\begin{split}
g\left(\boldsymbol{\Lambda}\right)=&
\frac{\pi^{m(m-1)}\text{etr}\left(-\Theta\right)}{\tilde{\Gamma}_m(n)\tilde{\Gamma}_m(m)\text{det}^n(\boldsymbol{\Psi})}
 \prod_{k=1}^{m}\lambda_k^{n-m}\prod_{k<l}\left(\lambda_k-\lambda_l\right)^2\\
& \times
\sum_{k,t=0}^{\infty}\sum_{\kappa,\tau;\phi\in\kappa.\tau}
\frac{C_{\phi}^{\kappa,\tau}\left(-\boldsymbol{\Psi}^{-1},\boldsymbol{\Theta}\boldsymbol{\Psi}^{-1} \right)
C_{\phi}^{\kappa,\tau}\left(\boldsymbol{\Lambda},\boldsymbol{\Lambda}\right)}{k!t![n]_{\tau}C_{\phi}\left(\mathbf{I}_m\right)}
\end{split}
\end{equation}
\end{cor}
\noindent where $\boldsymbol{\Lambda}$ is the  diagonal matrix having
eigenvalues of $\mathbf{W}$ along the main diagonal,    $C_{\phi}^{\kappa,\tau}\left(\cdot,\cdot \right)$ denotes  an invariant polynomial with two matrix arguments  \cite{davis1979invariant, davis1981construction} and the complex hypergeometric coefficient $[n]_{\tau}$ is defined as
$
[n]_{\tau}=\prod_{j=1}^{m}\left(n-j+1\right)_{\tau_j}
$
 in which $\tau=\left(\tau_1,\tau_2,....,\tau_m\right)$ is a partition of the integer $t$ into $m$ parts such that $\sum_{j=1}^m\tau_j=t $ and $\tau_1\geq\tau_2\geq .....\tau_m\geq 0$. Also, $(a)_k$ is the Pochhammer symbol given by 
$
(a)_k=\frac{\Gamma(a+k)}{\Gamma(a)}
$ with $(a)_0=1$.
\begin{figure*}
\begin{align*}
\mathcal{I}_k(x) =\sum_{p=0}^{\lfloor \frac{k}{2}\rfloor}\sum_{j=0}^{k-2p}
\frac{a_1(k,p,j)}{\left(x\sigma_2\right)^{j+p+2}}(j+p+1)! 
&\left[ \frac{{}_1F_1\left(1;c_{k,p,j}+3;x\sigma_1\right)}{(c_{k,p,j}+2)} - \sum_{i=0}^{j+p+1}\frac{(c_{k,p,j}+1)!}{(c_{k,p,j}+i+2)!}
\left(x\sigma_2\right)^{i}\right.\\& \left.\qquad \qquad \qquad \qquad \qquad \qquad \qquad \qquad \times {}_1F_1\left(i+1;c_{k,p,j}+i+3;x(\sigma_1-\sigma_2)\right)\right],
\end{align*}
\begin{align*}
\mathcal{J}_k(x) &=\sum_{p=0}^{\lfloor \frac{k}{2}\rfloor} \sum_{j=0}^{k-2p}\Biggl[a_2(k,p,j)\exp(-x\sigma_2) {}_1F_1\left(p+2;c_{k,p,j}+3;x\sigma_1\right)
{}_1F_1\left(p+2;j+p+3;x\sigma_2\right)\nonumber\\
& \qquad \left.+\sum_{l=0}^{p+1}\sum_{q=0}^{j+l}\frac{a_3(k,p,j,l,q)}{\left(x\sigma_2\right)^{j+l+1-q}} 
{}_1F_1\left(p+q+2;c_{k,p,j}+q+3;x(\sigma_1-\sigma_2)\right)-\sum_{l=0}^{p+1}\frac{a_4(k,p,j,l)}{\left(x\sigma_2\right)^{j+l+1}} {}_1F_1\left(p+2;c_{k,p,j}+3;x\sigma_1\right)\right]\nonumber
\end{align*}
\hrulefill
\end{figure*}

\subsection{Cumulative Distribution Function of the Maximum Eigenvalue}
Here,  to derive the c.d.f. of the maximum eigenvalue, the most straightforward method is to find the probability that the interval $[x,\infty)$ is free from the eigenvalues \cite{forrester2007eigenvalue, mehta2004random, zanella2009marginal, matthaiou2010condition}. As such, we may write
\begin{equation}
\begin{split}
F_{\lambda_{\text{max}}}(x) & =\Pr\left(\lambda_m<\lambda_{m-1}<\ldots<\lambda_1<x\right) \\
& =\int_{\mathcal{D}}g\left(\boldsymbol{\Lambda}\right) \rm d\lambda_1\ldots \rm d\lambda_m
\end{split}
\end{equation}
where $\mathcal{D}=\left\{0<\lambda_m<\ldots<\lambda_1<x\right\}$. This direct method,  however, is cumbersome because of  the invariant polynomials in eq (\ref{eigenpdf}). Despite their theoretical significance and frequent appearance in multivariate distribution theory, the invariant polynomials do not seem to admit simple forms in terms of the eigenvalues of its argument matrices even for the $2\times 2$ case \cite{davis1979invariant, davis1981construction}. To circumvent this difficulty, we will adopt an alternative approach based on integrating directly over matrix variate distribution, instead of the joint eigenvalue distribution \cite{constantine1963some, davis1979invariant, koev2006efficient, mathai1997jacobians, ratnarajah2003topics, dharmawansa2011extreme}. More precisely, we  write the c.d.f. of the maximum eigenvalue as
\begin{equation}
\label{initial}
F_{\lambda_{\text{max}}}(x)=\Pr\left(\mathbf{W}\prec x\mathbf{I}_m\right)
\end{equation}
which facilitates the use of the probability density function of $\mathbf{W}$ instead of its eigenvalue distribution.
Now we apply the change of variable $\mathbf{W}=x\mathbf{Z}$, where $\mathbf{Z}$ is Hermitian positive definite with ${\rm{d}} \mathbf{W}=x^{m^2}\rm d \mathbf{Z}$, in (\ref{initial}) to yield
\begin{equation}
\label{eq_maxeigint}
\begin{split}
F_{\lambda_{\text{max}}}(x) & =
\frac{x^{mn}\text{etr}\left(-\boldsymbol{\Theta}\right)}{\tilde{\Gamma}_m(n)\text{det}^n(\boldsymbol{\Psi})}\int_{\mathbf{0}\prec\mathbf{Z}\prec\;\mathbf{I}_m}
\text{det}^{n-m}(\mathbf{Z}) \\
& \qquad\text{etr}\left(-x\boldsymbol{\Psi}^{-1}\mathbf{Z}\right)
{}_0\tilde{F}_1\left(n;x\boldsymbol{\Theta}\boldsymbol{\Psi}^{-1}\mathbf{Z}\right)\rm d\mathbf{Z}.
\end{split}
\end{equation}
As shown in  \cite{ratnarajah2003topics} and \cite{dharmawansa2011extreme}, this integral does not admit simple form for arbitrary values of $m$ and $n$ even for rank one mean matrix. However, as we now show, it can be solved in terms of simple functions for the important configuration of $m=n=2$.\footnote{Although our general approach is valid even for $n>2$, here we focus on $n=2$ in view of obtaining a relatively not so complicated answer.} 

In the case of $m=n=2$, (\ref{eq_maxeigint}) simplifies to
\begin{equation}
\label{eq_maxeig22}
\begin{split}
F_{\lambda_{\text{max}}}(x)&=
\frac{x^{4}\text{etr}\left(-\eta\right)}{\pi\text{det}^2(\boldsymbol{\Psi})}\int_{\mathbf{0}\prec \mathbf{Z}\prec \;\mathbf{I}_2}
 \text{etr}\left(-x\boldsymbol{\Psi}^{-1}\mathbf{Z}\right) \\
 & \qquad
{}_0\tilde{F}_1\left(2;x\boldsymbol{\Theta}\boldsymbol{\Psi}^{-1}\mathbf{Z}\right)\rm d\mathbf{Z}
\end{split}
\end{equation}
where $\eta=\text{tr}\left(\boldsymbol{\Theta}\right)$. Observing the fact that the matrix $\boldsymbol{\Psi}$ is Hermitian positive definite having the eigen-decomposition $\boldsymbol{\Psi}=\mathbf{U}\boldsymbol{\Omega}\mathbf{U}^\dagger $, where $\mathbf{U}\in\mathbb{C}^{2\times 2}$ is unitary and $\boldsymbol{\Omega}=\text{diag}\{\omega_1,\omega_2\}$ with $\omega_1\geq \omega_2>0$, we can rewrite (\ref{eq_maxeig22}) with the help of  the transformation $\mathbf{Y}=\mathbf{U}^\dagger \mathbf{Z}\mathbf{U}$ as
\begin{equation}
\label{eq_maxtrformed}
\begin{split}
F_{\lambda_{\text{max}}}(x) & =
\frac{x^{4}}{\pi} \text{etr}\left(-\eta\right)(\sigma_1\sigma_2)^2\int_{\mathbf{0}\prec \mathbf{Y}\prec \;\mathbf{I}_2}
 \text{etr}\left(-x\boldsymbol{\Sigma}\mathbf{Y}\right) \\
 &\qquad
{}_0\tilde{F}_1\left(2;x\boldsymbol{\Sigma}\mathbf{U}^\dagger\boldsymbol{\Upsilon}^\dagger\boldsymbol{\Upsilon}\mathbf{U}\boldsymbol{\Sigma}\mathbf{Y}\right)\rm d\mathbf{Y}
\end{split}
\end{equation}
where $\boldsymbol{\Sigma}=\boldsymbol{\Omega}^{-1}=\text{diag}\{\sigma_1,\sigma_2\}$ and $0<\sigma_1\leq \sigma_2$. Observing the fact that $\boldsymbol{\Sigma}\mathbf{U}^\dagger\boldsymbol{\Upsilon}^\dagger\boldsymbol{\Upsilon}\mathbf{U}\boldsymbol{\Sigma}$ is Hermitian non-negative definite with rank one, it can be expressed via its eigen-decompsotion as
\begin{equation}
\label{eq_meandecom}
\boldsymbol{\Sigma}\mathbf{U}^\dagger\boldsymbol{\Upsilon}^\dagger\boldsymbol{\Upsilon}\mathbf{U}\boldsymbol{\Sigma}=\mu \boldsymbol{\alpha}\boldsymbol{\alpha}^\dagger
\end{equation}
where $\boldsymbol{\alpha}=(\alpha_{11}\;\; \alpha_{12})^T$ with $\alpha_{11},\alpha_{12}\in\mathbb{C}$, $\mu>0$ and $\boldsymbol{\alpha}^\dagger\boldsymbol{\alpha}=1$. This in turn gives the relation
\begin{equation}
\label{eq_mudef}
\mu=\text{tr}\left(\boldsymbol{\Theta\Psi}^{-1}\right).
\end{equation}
Therefore, we can write (\ref{eq_maxtrformed}) with the help of (\ref{eq_meandecom}) as
\begin{equation}
\begin{split}
F_{\lambda_{\text{max}}}(x)=&
\frac{x^{4}}{\pi} \text{etr}\left(-\eta\right)(\sigma_1\sigma_2)^2 \\
&\int_{\mathbf{0}\prec \mathbf{Y}\prec\;\mathbf{I}_2}
 \text{etr}\left(-x\boldsymbol{\Sigma}\mathbf{Y}\right)
{}_0\tilde{F}_1\left(2;x\mu \boldsymbol{\alpha}\boldsymbol{\alpha}^\dagger\mathbf{Y}\right)\rm d\mathbf{Y}.
\end{split}
\end{equation}
Further manipulation in this form is not desirable. However,  we expand the hypergeometric function with its equivalent zonal series expansion to yield \cite{james1964distributions}  
\begin{equation}
\label{eq_zonalexp}
\begin{split}
F_{\lambda_{\text{max}}}(x)=&
\frac{x^{4}}{\pi} \text{etr}\left(-\eta\right)(\sigma_1\sigma_2)^2\sum_{k=0}^\infty \sum_{\kappa}\frac{(x\mu)^k}{[2]_\kappa k!} \\
&\int_{\mathbf{0}\prec \mathbf{Y}\prec\;\mathbf{I}_2}
 \text{etr}\left(-x\boldsymbol{\Sigma}\mathbf{Y}\right)
C_\kappa\left(\boldsymbol{\alpha}\boldsymbol{\alpha}^\dagger\mathbf{Y}\right)\rm d\mathbf{Y}
\end{split}
\end{equation}
where $C_\kappa(\cdot)$ is the zonal polynomial in which $\kappa=(\kappa_1,\kappa_2)$ represents a partition of $k$ into not more than two parts such that $\kappa_1\geq \kappa_2\geq 0$ \cite{james1964distributions, takemura1984zonal}. Since the matrix $\boldsymbol{\alpha}\boldsymbol{\alpha}^\dagger\mathbf{Y}$ is rank one, we have $C_\kappa\left(\boldsymbol{\alpha}\boldsymbol{\alpha}^\dagger\mathbf{Y}\right)=0$ for all partitions $\kappa$ having more than one non-zero parts \cite{takemura1984zonal}. Therefore, the zonal polynomial degenerates to $C_\kappa\left(\boldsymbol{\alpha}\boldsymbol{\alpha}^\dagger\mathbf{Y}\right)=\left(\boldsymbol{\alpha}^\dagger\mathbf{Y}\boldsymbol{\alpha}\right)^k$. Hence we can simplify (\ref{eq_zonalexp}) to obtain
\begin{equation}
\label{cdf}
F_{\lambda_{\text{max}}}(x)=
\frac{x^{4}}{\pi} \text{etr}\left(-\eta\right)(\sigma_1\sigma_2)^2\sum_{k=0}^\infty \frac{(x\mu)^k}{(2)_k k!}\;Q_k(x)
\end{equation}
where
\begin{equation}
\label{eq_defq}
Q_k(x)
=\int_{\mathbf{0}\prec \mathbf{Y}\prec\;\mathbf{I}_2}
 \text{etr}\left(-x\boldsymbol{\Sigma}\mathbf{Y}\right)
\left(\boldsymbol{\alpha}^\dagger\mathbf{Y}\boldsymbol{\alpha}\right)^k\rm d\mathbf{Y}.
\end{equation}
This integral does not seem to have a simple solution in terms of simple functions according to the literature \cite{james1964distributions,ratnarajah2003topics}. Therefore, in the sequel, we demonstrate how to evaluate this integral in terms of hypergeometric  functions which in turn yield an exact expression for the c.d.f. of the maximum eigenvalue. 

The theorem below gives the c.d.f. of the maximum eigenvalue of a $2\times 2$ non-central Wishart matrix with two degrees of freedom.
\newtheorem{thm}{Theorem}
\begin{thm}\label{thm_maxeig}
Let $\mathbf{X}$ be a $2\times 2$ complex square matrix distributed as $\mathbf{X}\sim \mathcal{CN}_{2,2}\left(\boldsymbol{\Upsilon},\mathbf{I}_2\otimes\boldsymbol{\Psi}\right)$ with $\boldsymbol{\Upsilon}\in\mathbb{C}^{2\times 2}$ is rank one. Then the c.d.f. of the maximum eigenvalue $\lambda_{\text{max}}$ of the semi-correlated non-central complex Wishart matrix $\mathbf{W}=\mathbf{X}^\dagger\mathbf{X}$ is given by
\begin{align}
\label{eq_cdf}
F_{\lambda_{\text{max}}}(x)&=\left(\sigma_1\sigma_2\right)^2x^4\exp(-\sigma_1 x-\eta)\nonumber\\
&\qquad \qquad \times \sum_{k=0}^\infty
\frac{(x\mu)^k}{(k+1)!}\left[\mathcal{I}_k(x)+\mathcal{J}_k(x)\right]
\end{align}
where $\mathcal{I}_k(x)$ and $\mathcal{J}_k(x)$ are given on top of the page 
with
\begin{equation*}
\begin{array}{ll}
 & a_1(k,p,j)=\frac{|\alpha_{11}|^{2c_{k,p,j}}|\alpha_{12}|^{2(p+j)}}{j!p!(p+1)!(c_{k,p,j}-p)!} \\
 & a_2(k,p,j)=\frac{(p+1)|\alpha_{11}|^{2c_{k,p,j}}|\alpha_{12}|^{2(p+j)}}{(j+p+2)!(c_{k,p,j}+2)!}\\
& a_3(k,p,j,l,q)=\frac{(-1)^l(p+1)(j+l)!|\alpha_{11}|^{2c_{k,p,j}}|\alpha_{12}|^{2(p+j)}}{j!\;l!(p+1-l)!(c_{k,p,j}+2)!} \\
 & a_4(k,p,j,l)=\frac{(-1)^l(j+l)!(p+q+1)! |\alpha_{11}|^{2c_{k,p,j}}|\alpha_{12}|^{2(p+j)}}{j!l!p!q!(p+1-l)! (c_{k,p,j}+q+2)!}
\end{array}
\end{equation*}
and $c_{k,p,j}=k-p-j$. Also, ${}_1F_1(\cdot;\cdot;\cdot)$ denotes the confluent hypergeometric function of the first kind, $\boldsymbol{\Theta}=\boldsymbol{\Psi}^{-1}\boldsymbol{\Upsilon}^\dagger\boldsymbol{\Upsilon}$, $\boldsymbol{\Psi}=\mathbf{U}\boldsymbol{\Omega}\mathbf{U}^\dagger $,
$\mu=\text{tr}\left(\boldsymbol{\Theta\Psi}^{-1}\right)$, $\eta=\text{tr}\left(\boldsymbol{\Theta}\right)$, 
$\boldsymbol{\Sigma}=\boldsymbol{\Omega}^{-1}=\text{diag}\{\sigma_1,\sigma_2\}$, 
$\boldsymbol{\Sigma}\mathbf{U}^\dagger\boldsymbol{\Upsilon}^\dagger\boldsymbol{\Upsilon}\mathbf{U}\boldsymbol{\Sigma}=\mu \boldsymbol{\alpha}\boldsymbol{\alpha}^\dagger$, and 
$\boldsymbol{\alpha}=(\alpha_{11}\; \alpha_{12})^T$. 
\end{thm}
\begin{proof}
See Appendix A.
\end{proof}

\begin{figure}[t]
\centering
\includegraphics[width=0.4\textwidth]{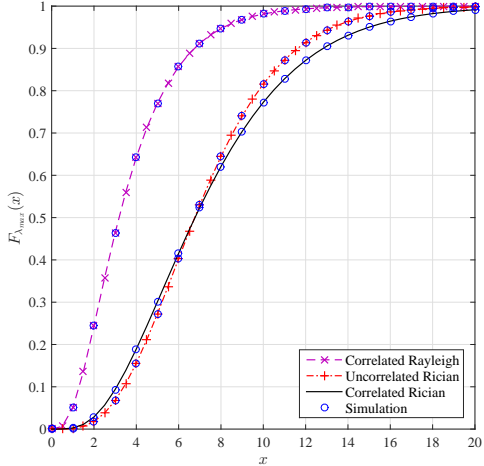}
\caption{Outage probability vs normalized SNR threshold $\gamma_{\text{th}}/\bar{\gamma}$ for correlated Rayleigh, uncorrelated/correlated Rician without $K$ normalization. 
}
\label{fig:cdf_k}
\end{figure}

It is remarkable that the above c.d.f. expression depends on, among other parameters, the components of the eigenvector of the rank one matrix $\boldsymbol{\Sigma}\mathbf{U}^\dagger\boldsymbol{\Upsilon}^\dagger\boldsymbol{\Upsilon}\mathbf{U}\boldsymbol{\Sigma}$, in contrast to corresponding expressions for correlated Rayleigh and uncorrelated Rician matrices which do not depend on the eigenvectors \cite{zanella2009marginal}. This represents the joint effect of correlation and mean. As a simple sanity  check of Theorem 1,  Fig. \ref{fig:cdf_k} shows the comparison of theoretical c.d.f. expression (\ref{eq_cdf}) and simulation results for the following parameters: 
\begin{align}
\label{eq_num_matrix}
\mathbf{\Upsilon}&=\left(\begin{array}{cc}
   1.0000+0.0000i \quad0.3624-0.9320i \\
   0.8878-0.4603i \quad-0.1073-0.9942i
\end{array}\right)\nonumber\\
\boldsymbol{\Psi}&=\left(\begin{array}{cc}
   1.0000+0.0000i \quad-0.3731-0.4902i \\
  -0.3731+0.4902i \quad1.0000+0.0000i
\end{array}\right).
\end{align}
It is also noteworthy that the infinite series has been truncated to a maximum of $15$ terms; thereby demonstrating a fast convergence rate for each case.

\section{Performance of $2\times 2$ MIMO Beamforming}
\label{sec:peranalysis}
To emphasize the utility of Theorem 1, here we focus on the performance of  $2\times 2$ MIMO MRC  over  a correlated Rician fading channel. In particular, we analyze an important performance metric - outage probability. 

Consider the following $2\times 2$ MIMO Rician channel model \cite{jin2008mimo, wu2016asymptotic}
\begin{equation}
	\label{channel_model_decom}
	\mathbf{H}=\sqrt{\frac{K}{K+1}}{\mathbf{\bar{H}}}+\sqrt[]{\frac{1}{K+1}}\mathbf{H_{\it sc}T}^{\frac{1}{2}}
	\end{equation}
where $ \mathbf{\bar{H}} \in \mathbb{C}^{2\times 2}$ is the deterministic component, $ \mathbf{H_{\it sc}} \sim \mathcal{CN}_{2,2}(\mathbf{0}_{2 \times 2},\mathbf{I}_{2}\otimes {\bf I}_{2}) $ represents the Rayleigh random component, and $\mathbf{T}\in\mathbb{C}^{2\times 2}\succ \mathbf{0}$ is the transmit correlation matrix with $K$ being the Rician factor. Moreover, we adapt the common normalization used in the literature to suit our requirement as $
\text{tr}\left(\mathbf{\bar H}^\dagger\mathbf{\bar H}\right)=4$ and $
\text{tr}\left(\mathbf{T}\right)=2$. It is noteworthy that in the presence of a strong LoS path between the transmitter and receiver (i.e., Rician fading), the rank of the matrix $\mathbf{\bar H}$ degenerates to one \cite{jin2008mimo, wu2016asymptotic}. Therefore, here we focus on the rank-1 $\bar{\mathbf{H}}$ case only.
Since $\mathbf{H}$ is a complex Gaussian random matrix, following Definition \ref{Wishart_Def}, $\mathbf{W}=\mathbf{H}^\dagger\mathbf{H}$ is  correlated complex non-central Wishart distributed.  Therefore, the corresponding covariance and the non-centrality parameter matrices can be written, respectively, as $\mathbf{\Psi}=(\frac{1}{K+1})\mathbf T$ and $\mathbf{\Theta=\Psi^{-1}\Upsilon^{\dagger}\Upsilon}=K{\mathbf T}^{-1}{\bf \bar{H}^\dagger\bar{H}}$.

Now consider a point-to-point MIMO link  with two transmit and two receive antennas. The received information vector $\mathbf{r}\in \mathbb{C}^{2\times 1}$ is given by
\begin{align*}
\mathbf{r}=\sqrt{P} \mathbf{H}\mathbf{w}s+\mathbf{n}
\end{align*}
where the channel $\mathbf{H}\in\mathbb{C}^{2\times 2}$ is given by (\ref{channel_model_decom}) with $\bar{\mathbf{H}}$ being {\it rank-1}, $P$ is the transmit power, $\mathbf{w}\in\mathbb{C}^{2\times 1}$ is the beamforming vector with $||\mathbf{w}||^2=1$, $s\in \mathbb{C}$ is the transmitted symbol with $\mathbb{E}\left\{|s|^2\right\}=1$, and $\mathbf{n}\in\mathbb{C}^{2\times 1}$ is the additive white Gaussian noise (AWGN) vector with zero mean and $N_0\mathbf{I}_2$ covariance matrix. We assume that both transmitter and the receiver have perfect instantaneous CSI.

A MIMO-MRC receiver determines $\mathbf{w}$ such that the received SNR
\begin{align}
\label{MRC_SNR}
\gamma=\bar{\gamma} \mathbf{w}^\dagger \mathbf{H}^\dagger \mathbf{H} \mathbf{w}
\end{align}
is maximized. Here $\bar{\gamma}=P/N_0$ denotes the transmit SNR. It is well documented that the vector $\mathbf{w}$ which maximizes (\ref{MRC_SNR}) is the leading eigenvector of $\mathbf{H}^\dagger\mathbf{H}$  \cite{dighe2003analysis}, \cite{kang2003comparative,kang2003impact,kang2003largest}. Therefore, the maximum received SNR is given by
\begin{align}
\label{eq_maxSNR}
\gamma=\bar{\gamma} \lambda_{\max}
\end{align}
where $\lambda_{\max}$ is the leading eigenvalue of $\mathbf{H}^\dagger\mathbf{H}$. 
This  clearly demonstrates that the performance of MIMO-MRC is tightly coupled with  the statistics of $\lambda_{\max}$. We next  focus on evaluating the outage probability of this system.

\subsection{Outage Probability}
\label{sec:outprobability}
The outage probability characterizes the quality of service (QoS) provided by the system, and is a more generic  performance measure of user experience. 
It is formally defined as the probability of $\gamma$ falls below a certain threshold value, $\gamma_{\text{th}}$, which determines the minimum SNR level for  satisfactory reception. Following Theorem \ref{thm_maxeig}, the outage probability can be written as
\begin{align}
\label{eq_outage}
P_{\text{out}}\left(\frac{\gamma_{\text{th}}}{\bar{\gamma}}\right) & =\Pr\left\{\gamma<\gamma_{\text{th}}\right\}=\Pr\left\{\lambda_{\max}<\frac{\gamma_{\text{th}}}{\bar{\gamma}}\right\} \nonumber \\
& =F_{\lambda_{\max}}\left(\frac{\gamma_{\text{th}}}{\bar{\gamma}}\right)
\end{align}
where $F_{\lambda_{\max}}(x)$ is given by (\ref{eq_cdf}) with the following re-parameterization:
\begin{align*}
\mu=K(K+1) \text{tr}\left(\mathbf{T}^{-1}\bar{\mathbf{H}}^\dagger \bar{\mathbf{H}}\mathbf{T}^{-1}\right),\; 
\eta=K\text{tr}\left(\mathbf{T}^{-1}\bar{\mathbf{H}}^\dagger \bar{\mathbf{H}}\right),
\end{align*}
and $\boldsymbol{\alpha}=\left(\alpha_{11}\; \alpha_{12}\right)^T$ denotes the leading eigenvector of the rank-1 matrix  $\boldsymbol{\Lambda}^{-1}\mathbf{U}^\dagger\bar{\mathbf{H}}^\dagger \bar{\mathbf{H}}\mathbf{U}\boldsymbol{\Lambda}^{-1}$. Here $\mathbf{U}$ and $\boldsymbol{\Lambda}$ are related to $\mathbf{T}$ through the eigen-decomposition $\mathbf{T}=\mathbf{U}\boldsymbol{\Lambda}\mathbf{U}^\dagger$ and $\boldsymbol{\Sigma}=\text{diag}(\sigma_1, \sigma_2)=(K+1)\boldsymbol{\Lambda}^{-1}$.

\begin{figure}[t]
\centering
\includegraphics[width=0.4\textwidth]{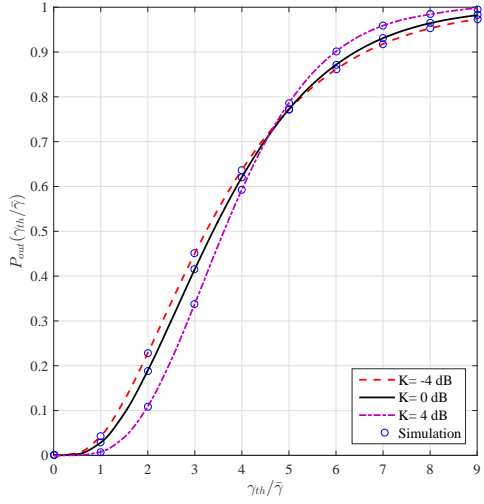}
\caption{Outage probability vs normalized SNR threshold $\gamma_{\text{th}}/\bar{\gamma}$ for different $K$ values.}
\label{fig:cdf_k_2}
\end{figure}

 The accuracy of (\ref{eq_outage}) is verified in Fig. \ref{fig:cdf_k_2}, which plots the outage probability versus the normalized SNR threshold $\gamma_{\text{th}}/\bar{\gamma}$ for different Rician factors  (different $K$ values). For numerical simulation purposes,  we have used $\mathbf{T}=\boldsymbol{\Psi}$ and $\bar{\mathbf{H}}=\boldsymbol{\Upsilon}$ matrices given in (\ref{eq_num_matrix}). Note that for  these analytical curves, the infinite sum (\ref{eq_cdf}) has been truncated to a maximum of $15$ terms,  thereby demonstrating a fast convergence rate for each case. The  figure displays a close match between the simulations and analytical results, which verifies the accuracy of the analytical derivations. A counterintuitive  trend is also visible - namely the outage does not uniformly decreases with the strength of the LoS component for all $\bar{\gamma}$ values. In particular, in the large $\bar{\gamma}$ regime (i.e., low $\gamma_{\text{th}}/\bar{\gamma}$ regime), outage improves with the increasing strength of the LoS component. In contrast, the opposite trend can be observed in the low $\bar{\gamma}$ regime. Therefore, when the transmit SNR is high, the outage is benefited by having a strong LoS component, whereas in the low transmit SNR regime, a rich scattering environment certainly helps improve the outage.  Thus,  from  an outage point of view and counterintuitively,   it is not always beneficial to have stronger LoS links.  
 
\subsection{The Effect of $K$}
We now investigate the effect of Rician factor $K$ on the outage probability. In particular, we focus on the large $K$ regime. The  following proposition characterizes the outage probability as $K\to \infty$:
\newtheorem{pro}{Proposition}
\begin{pro}\label{pro_asy_outage}
Let $\bar{\mathbf{H}^\dagger} \bar{\mathbf{H}}=4\mathbf{v}\mathbf{v}^\dagger$, where $||\mathbf{v}||=1$. Then we have, as $K\to\infty$
\begin{align}
\sqrt{\frac{K}{8\bar{\gamma}^2\;\mathbf{v}^{\dagger}\mathbf{T}\mathbf{v}}}\left(\gamma-4\bar{\gamma}\right)\xrightarrow{\mathcal{D}} \mathcal{N}(0,1)
\end{align}
where $\xrightarrow{\mathcal{D}} \mathcal{N}(0,1)$ denotes the convergence in distribution to a standard normal random variable.
\end{pro}
\begin{proof}
Omitted due to space limitation.
\end{proof}
An immediate consequence of Proposition \ref{pro_asy_outage} is 
\begin{align}
\label{eq_conv_inter}
\lim_{K\to \infty}\Pr\left\{\sqrt{\frac{K}{8}}\left(\gamma-4\bar{\gamma}\right)\leq \gamma_{\text{th}}\right\}=\Phi\left(\frac{\gamma_{\text{th}}/\bar{\gamma}}{\sqrt{\mathbf{v}^{\dagger}\mathbf{T}\mathbf{v}}}\right)
\end{align}
where $\Phi(z)$ is the c.d.f of the standard normal random variable. Clearly, once properly centered and scaled, the asymptotic outage depends on the channel mean and correlation through the positive definite quadratic form $\mathbf{v}^\dagger \mathbf{T}\mathbf{v}$. It is noteworthy that this quadratic form depends on $\mathbf{T}$ but not its inverse. Therefore, the asymptotic outage expression remains valid even when $\mathbf{T}$ is {\it rank deficient}. Since we are interested in outage probability, we may use (\ref{eq_conv_inter}) to approximately characterize it, for sufficiently large $K$ values, as 
\begin{align}
\label{eq_approx_largeK}
P^{\text{Large}\;K}_{\text{out}}\left(\frac{\gamma_{\text{th}}}{\bar{\gamma}}\right)\approx
\Phi\left(\frac{\gamma_{\text{th}}/\bar{\gamma}-4}{\sqrt{\displaystyle 8\mathbf{v}^\dagger \mathbf{T} \mathbf{v}/K}}\right).
\end{align}
It is interesting to observe that the parameter $K$ is completely decoupled from the other parameters in \eqref{eq_approx_largeK}.

The effect of large $K$ on the outage probability is depicted in Fig. \ref{fig:large_k_approx_a}. As we  already know, for  $K\to \infty$, the maximum eigenvalue tends  to concentrate around $4$ (i.e, $\gamma$  around $4\bar{\gamma}$). Therefore, as  the figure shows, the outage curves bend more sharply at the critical point and ultimately converges to the vertical barrier at normalized SNR of  $4$. 
\begin{figure*}[t]
	\centering
	\subfloat[Formation of an outage barrier at $\gamma_{\text{th}}/\bar{\gamma}=4$.]{
		\label{fig:large_k_approx_a}
		\includegraphics[width=0.4\textwidth]{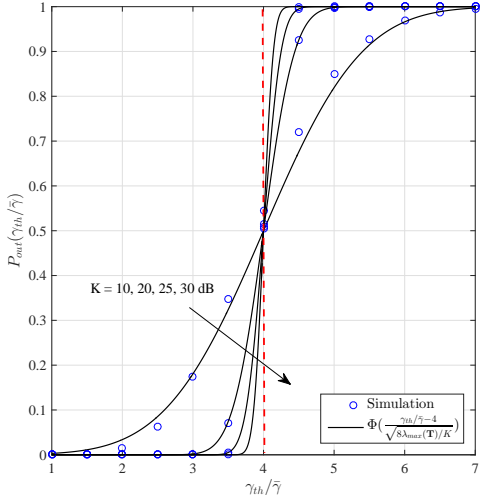}}      
        \subfloat[The joint effect of $\mathbf{T}$ and $\mathbf{v}$ on outage for large $K$. We have used $K=30$ dB for numerical evaluations.]{
		\label{fig:large_k_approx_b}
		\includegraphics[width=0.4\textwidth]{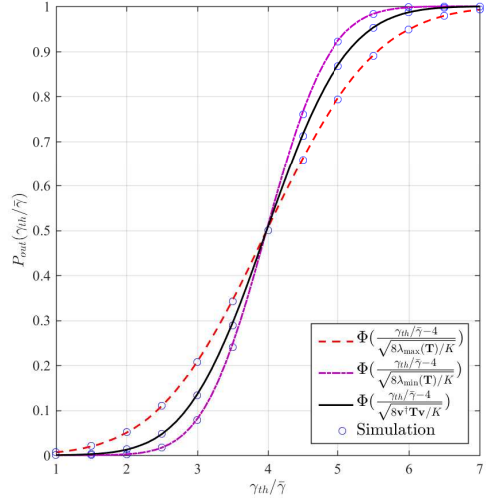}} 
        \caption{Asymptotically Gaussian behavior of outage as $K\to\infty$.}
	\label{f_FdHdb}
\end{figure*}

Let us now investigate the joint effect of $\mathbf{v}$ and $\mathbf{T}$ on the outage probability. To this end,  we first note that $ \mathbf{v}^\dagger\mathbf{T}\mathbf{v}$ is maximized when $\mathbf{v}$ aligns with the leading eigenvector of $\mathbf{T}$, whereas $ \mathbf{v}^\dagger\mathbf{T}\mathbf{v}$ is minimized when $\mathbf{v}$ aligns with the least eigenvector of $\mathbf{T}$. Therefore, when $\lambda_{\max}(\mathbf{T})$ and $\lambda_{\min}(\mathbf{T})$ represent, respectively, the maximum and the minimum eigenvalue of $\mathbf{T}$, we have the inequality
$
\lambda_{\min}(\mathbf{T})\leq \mathbf{v}^\dagger \mathbf{T} \mathbf{v}\leq
\lambda_{\max}(\mathbf{T}).
$
Thus, keeping in mind that $\Phi(z)$ is a monotonically increasing function, we have from (\ref{eq_approx_largeK}), for $\frac{\gamma_{\text{th}}}{\bar{\gamma}}<4$, 
$\Phi\left( \frac{\gamma_{\text{th}}/\bar{\gamma}-4}{\sqrt{ 8\lambda_{\min}(\mathbf{T})/K}}\right)<
\Phi\left( \frac{\gamma_{\text{th}}/\bar{\gamma}-4}{\sqrt{ 8\mathbf{v}^\dagger \mathbf{T} \mathbf{v}/K}}\right)
<\Phi\left( \frac{\gamma_{\text{th}}/\bar{\gamma}-4}{\sqrt{ 8\lambda_{\max}(\mathbf{T})/K}}\right)$, and for $ \frac{\gamma_{\text{th}}}{\bar{\gamma}}<4$, $\Phi\left( \frac{\gamma_{\text{th}}/\bar{\gamma}-4}{\sqrt{ 8\lambda_{\min}(\mathbf{T})/K}}\right)>
\Phi\left( \frac{\gamma_{\text{th}}/\bar{\gamma}-4}{\sqrt{ 8\mathbf{v}^\dagger \mathbf{T} \mathbf{v}/K}}\right)
>\Phi\left( \frac{\gamma_{\text{th}}/\bar{\gamma}-4}{\sqrt{ 8\lambda_{\max}(\mathbf{T})/K}}\right)$
with the critical normalized SNR threshold $\gamma_{\text{th}}/\bar{\gamma}=4$ satisfying the equality among the three quantities. Clearly, at the critical normalized SNR value, the outage curves undergo a phase transition. Therefore, for sufficiently large $K$, above the critical threshold, it is beneficial to have a $\mathbf{v}$ vector aligned with the leading eigenvector of the $\mathbf{T}$ from the outage point of view. In contrast, below the critical threshold, a $\mathbf{v}$ vector aligned with the least eigenvector of $\mathbf{T}$ improves the outage performance. 

Figure \ref{fig:large_k_approx_b}  numerically illustrates the above insights. As shown in the figure, below the critical threshold (i.e., high transmit SNR regime), alignment of the mean vector with the least eigenvalue of the transmit correlation matrix outperforms the other mode of alignments. Also the  inverse behavior is depicted above the critical threshold (i.e., low transmit SNR regime).

\section{Conclusion}
This paper has focused on characterizing $2\times 2$ MIMO-MRC system  over correlated  Rician fading  and a  strong LoS component (rank-$1$ mean) with respect to outage probability. 
We first derive a new expression for the c.d.f. of the maximum eigenvalue of a $2\times 2$ correlated non-central Wishart matrix with rank-$1$ non-centrality parameter. This expression, which in turn facilitates the derivation of the outage, contains a fast converging infinite series of functions. 
Our analysis demonstrates that as the Rician factor ($K$) grows large, the outage can be approximately characterized by the c.d.f. of a certain Gaussian random variable. Interestingly,  our analysis shows that  a strong LoS path is not always beneficial from the outage perspective.

\appendices

\section{Proof of Theorem \ref{thm_maxeig}}
Here we provide an outline of the main proof due to space limitations.

For clarity let us recall (\ref{eq_defq})
\begin{equation}
\label{eq_appendefq}
Q_k(x)
=\int_{\mathbf{0} \prec \mathbf{Y}\prec\;\mathbf{I}_2}
 \text{etr}\left(-x\boldsymbol{\Sigma}\mathbf{Y}\right)
\left(\boldsymbol{\alpha}^\dagger\mathbf{Y}\boldsymbol{\alpha}\right)^k\rm d\mathbf{Y}.
\end{equation}
 Clearly, this integral is performed over the space of Hermitian positive definite matrices $\mathbf{Y}\in\mathbb{C}^{2\times 2}$ such that $\mathbf{I}_2-\mathbf{Y}\succ\mathbf{0}$. This more complicated integral region in turn makes the above matrix integral intractable with the available tabulated results in the literature. Therefore, we now take a closer look at this particular integral region. To this end we first parameterize the Hermitian positive definite matrix $\mathbf{Y}$ as
 \begin{equation}
 \mathbf{Y}=\left(\begin{array}{cc}
 y_{11} & y_{12}\\
 y^*_{12} & y_{22}
 \end{array}\right)
 \end{equation}
where $y_{11},y_{22}>0$, $y_{12}\in\mathbb{C}$ and $y_{11}y_{22}-|y_{12}|^2>0$. Since we require $\mathbf{I}_2-\mathbf{Y}\succ\mathbf{0}$ to be satisfied, the above parameters should also fulfill
\begin{equation}
(1-y_{11})(1-y_{22})-|y_{12}|^2>0,\;\; y_{11},y_{22}\in(0,1).
\end{equation}
Now keeping in mind the representation $y_{12}=\Re y_{12}+\mathfrak{i}\Im y_{12}$, we can combine the above results to write the integration region corresponding to $\mathbf{0}\prec \mathbf{Y}\prec \;\mathbf{I}_2$ as
\begin{align}
\mathcal{R}=\mathcal{R}_1\cup \mathcal{R}_2
\end{align}
where
\begin{align*}
\mathcal{R}_1&=\left\{(y_{11},y_{22},\Re y_{12},\Im y_{12} ):
y_{11}+y_{22}<1,\right.\\
& \left.\qquad \qquad \quad |y_{12}|^2<y_{11}y_{22},\;y_{11},y_{22}\in(0,1) \right\}\\
\mathcal{R}_2&=\left\{(y_{11},y_{22},\Re y_{12},\Im y_{12} ): y_{11}+y_{22}>1,\right.\\
& \left.\qquad |y_{12}|^2<(1-y_{11})(1-y_{22}),\; y_{11},y_{22}\in(0,1)\right\}.
\end{align*}
Capitalizing on the above facts we can simplify the matrix integral in (\ref{eq_appendefq}) to yield the scalar form
\begin{align}
\label{eq_Qdecom}
Q_k(x)=
P_1(k,x)+P_2(k,x)
\end{align}
where
\begin{align*}
P_1(k,x)& =\int_{\mathcal{R}_1}
 \exp\left(-x\sigma_1y_{11}-x\sigma_2y_{22}\right)\nonumber\\
& \qquad  \times 
\left(|\alpha_{11}|^2y_{11}+|\alpha_{12}|^2y_{22}+2\Re[ \alpha^*_{11}\alpha_{12}y_{12}]\right)^k
{\rm d} \mathbf{Y}\\
P_2(k,x)& =\int_{ \mathcal{R}_2}
 \exp\left(-x\sigma_1y_{11}-x\sigma_2y_{22}\right)\nonumber\\
& \qquad  \times 
\left(|\alpha_{11}|^2y_{11}+|\alpha_{12}|^2y_{22}+2\Re[ \alpha^*_{11}\alpha_{12}y_{12}]\right)^k
{\rm d} \mathbf{Y}
\end{align*}
and we have used the differential relation ${\rm d} \mathbf{Y}={\rm d}y_{11}{\rm d}y_{22} {\rm d}\Re y_{12} {\rm d}\Im y_{12}$. The final answer follows by evaluating the above four-dimensional integrals.



\end{document}